\documentclass[runningheads]{llncs}
\usepackage[T1]{fontenc}
\usepackage{graphicx}
\def\tr

\usepackage{amsmath}
\usepackage{amsfonts}
\usepackage{amssymb}
\usepackage{booktabs}
\usepackage{calc}
\usepackage[sort]{cite}
\usepackage{hyperref}
\usepackage{listings}
\usepackage{mathtools}
\usepackage{stmaryrd}
\usepackage{subcaption}
\usepackage{tikz}
\usepackage{twoopt}
\usepackage{upgreek}
\usepackage{upquote}
\usepackage{xcolor}
\usepackage{xspace}

\newcommand{\typedsym}{{:}}

\newcommand{\recv}{\text{\upshape?}}

\newcommand{\typeannotx}[1]{{\color{teal}\typedsym#1}}

\newcommand{\recvx}[2][]{\recv(#2\IfStrEq{#1}{}{}{\typeannotx{#1}})}

\newcommand{\dict}{\smash{\mathfrak{D}}}

\usepackage{nicefrac}
\usepackage{centernot}

\newcommand{\nterm}[1][\dict]{{{}\mathrel{\centernot{\downarrow}}}}
\newcommand{\ntermsym}[1][\dict]{{\centernot{\downarrow}}}

\newcommand{\inxx}[3][]{#2(#3\ifthenelse{\equal{#1}{}}{}{{{\color{black}\isa#1}}})}

\newcommandtwoopt{\loopx}[3][][]{\loopsym[#1]\ifthenelse{\equal{#2}{}}{}{(#2)}\mskip\medmuskip#3}
\newcommand{\loopsym}[1][]{\smash{\mathbf{loop}}_{#1}}
\newcommandtwoopt{\recur}[2][][]{\mathbf{recur}_{#1}\ifthenelse{\equal{#2}{}}{}{(#2)}}

\newcommand{\isa}{\mspace{1mu}\typedsym\mspace{1mu}}

\newcommand{\bin}[1][]{\mskip\medmuskip{\binsym}\mskip\medmuskip}
\newcommand{\binsym}[1][]{{\oplus}}

\newcommandtwoopt{\rollsym}[2][][]{{\prec_{#1}^{#2}}}

\lstdefinelanguage{scrib}{%
  keywords=[1]{choice,at,or,local,from,to,rec,continue,par,and,end,data,global,protocol,do,role},
  keywordstyle=[1]\color{purple},
  keywords=[2]{Accept,Confirm,Reject,Propose,Int},
  keywordstyle=[2]\color{teal},
  morecomment=[l]{//},
  morecomment=[s]{/*}{*/},
  commentstyle=\color{gray},
  mathescape=true,
}

\lstdefinelanguage{scala}{%
  alsoletter={@},
  keywords=[1]{new, if, then, else, class, return, case, match, import, type, val, trait, while, do, var, def, extends, throw, true, false, throw, private},
  keywordstyle=[1]\color{purple},
  keywords=[2]{Cup, Cap, Consistent, Loop, End, Recur, Map, Fold, Head, Last, Proj, Com, Merg, Tuple, S, T, T1, Tn, Ti, F, Z, X, R, U, P, Q, Login, Password, Auth, Cancel, Role, Type, G, G1, Gn, Gi, D, E, L, L1, L2, S@B, S@S, Delegate, T@S, Merge, Send, Recv, Accept, Reject, Propose, Confirm, SendCallback, RecvCallbacks, Substitute, App, Network, UseOnce, State1, State2, State3, State4, State5, State6, S1, S2, S3, S4, S5, S6, Ok, Quit, Thread, String, Int, Boolean, IntOrBoolean, Convert, Date, LocalSession, GlobalSession, Project, Local, Proj, Global},
  keywordstyle=[2]\color{teal},
  morecomment=[l]{//},
  morecomment=[s]{/*}{*/},
  commentstyle=\color{gray},
  morestring=[b]",
  stringstyle=\color{black},
  mathescape=true,
}

\lstdefinelanguage{dcj}{%
  keywords=[1]{defrole, defthread, defroles, defsession, alt, cat, cat*, par, -->, -->>, ->>, *},
  keywordstyle=[1]\color{purple},
  keywords=[2]{},
  keywordstyle=[2]\color{teal},
  morecomment=[l]{;},
  commentstyle=\color{gray},
  mathescape=true,
  alsoletter={>,-,*}
}

\lstdefinelanguage{clj}{%
  keywords=[1]{thread, chan, >!!, <!!, do, do-alts, partial, deliver, deref, promise, alt!!, alts!!, alts!!-with-liveness-check, fn, link, monitor, session, defn, def, let, if, true?, get, true, loop, recur, /, =, not=, println, nil, swap!, inc, dec, throw, ex-info, atom, nil?, second},
  keywordstyle=[1]\color{teal},
  keywords=[2]{},
  keywordstyle=[2]\color{teal},
  morecomment=[l]{;},
  morecomment=[s]{/*}{*/},
  commentstyle=\color{gray},
  mathescape=true,
  alsoletter={<,>,!,/,=,-,?}
}

\newcommand{\dcjm}{\lstinline[language=dcj,mathescape=true,basicstyle=\ttfamily\small\upshape]}
\newcommand{\clj}{\lstinline[language=clj,mathescape=false,basicstyle=\ttfamily\small\upshape]}
\newcommand{\cljm}{\lstinline[language=clj,mathescape=true,basicstyle=\ttfamily\small\upshape]}

,stringstyle={\color{teal}}]}

\usepackage{environ}
\NewEnviron{listgather*}{%
	\begin{gather*}
		\hspace{24.5pt}\hbox to \linewidth {\hfil$\begin{gathered}[t]
			\BODY
		\end{gathered}$\hfil}
	\end{gather*}%
}

\usetikzlibrary{decorations.pathreplacing,calligraphy}

\usetikzlibrary{calc}

\tikzstyle{state} = [inner sep=.5mm, outer sep=.25mm, circle, fill, line width=0pt]
\tikzstyle{final} = [inner sep=.75mm, outer sep=.25mm, circle, fill=none, draw]
\tikzstyle{trans} = [-stealth, rounded corners=1mm]
\tikzstyle{bitrans} = [stealth-stealth, rounded corners=1mm]
\tikzstyle{label} = [inner sep=0pt, line width=0pt, font=\scriptsize]
\tikzstyle{label-left} = [label, anchor=base east, xshift=-1mm, yshift=-.25ex]
\tikzstyle{label-right} = [label, anchor=base west, xshift=1mm, yshift=-.25ex]

\usetikzlibrary{arrows.meta}
\usetikzlibrary{decorations.pathmorphing}

\tikzstyle{object} = [font=\footnotesize]
\tikzstyle{wred} = [double, -{Implies}]
\tikzstyle{wbisim} = [double, decorate, decoration={snake, amplitude=.25mm, segment length=2mm}]

\let\primenosmash\prime
\let\astnosmash\ast
\let\dagnosmash\dag
\let\ddagnosmash\ddag
\let\Snosmash\S
\let\Pnosmash\P
\renewcommand{\prime}{{\smash{\primenosmash}}}
\renewcommand{\ast}{{\smash{\astnosmash}}}
\renewcommand{\dag}{{\smash{\dagnosmash}}}
\renewcommand{\ddag}{{\smash{\ddagnosmash}}}
\renewcommand{\S}{{\smash{\Snosmash}}}
\renewcommand{\P}{{\smash{\Pnosmash}}}

\newcommand{\dfracfrac}[2]{\text{%
  \rlap{\ensuremath{\dfrac{\phantom{#1}}{#2}}}%
  \raisebox{2pt}{\ensuremath{\dfrac{#1}{\phantom{#2}}}}%
}}

\newcommand{\rulexxx}[3]{\dfrac{#2}{#3}\ifthenelse{\equal{#1}{}}{}{\  \text{\upshape(\hypertarget{\detokenize{#1}}{{#1}})}}}
\newcommand{\corulexxx}[3]{\dfracfrac{\begin{gathered}#2\end{gathered}}{\begin{gathered}#3\end{gathered}}\ifthenelse{\equal{#1}{}}{}{\, \text{\upshape\scriptsize[\hypertarget{\detokenize{#1}}{\textsc{#1}}]}}}

\newcommandtwoopt{\BLOCK}[3][{[}][{]}]{%
	\begingroup%
				\left#1\begin{gathered}#3\end{gathered}\right#2%
	\endgroup%
}

\newcommandtwoopt{\BLOCKt}[3][{[}][{]}]{%
	\begingroup%
				\left#1\begin{gathered}[t]#3\end{gathered}\right#2%
	\endgroup%
}

\NewEnviron{hide}{}
\let\oldproof\proof
\let\oldendproof\endproof

\newcommand{\hideproofs}{%
	\let\proof\hide%
	\let\endproof\endhide}
	
\newcommand{\showproofs}{%
	\let\proof\oldproof%
	\let\endproof\oldendproof}

\ifdefined\forceshowproofs
  \let\hideproofs\showproofs
\fi
\ifdefined\forcehideproofs
  \let\showproofs\hideproofs
\fi

\begin{document}

\title{Discourje: Run-Time Verification of Communica-\\tion Protocols in Clojure -- Live at Last\ifdefined\tr\\(Technical Report)\fi}
\titlerunning{Discourje: Live at Last}
\author{Sung-Shik Jongmans}
\institute{Open University of the Netherlands, Heerlen, the Netherlands}
\maketitle

\begin{abstract}
Multiparty session typing (MPST) is a formal method to make concurrent
programming simpler. The idea is to use type checking to automatically prove
safety (protocol compliance) and liveness (communication deadlock freedom) of
implementations relative to specifications. Discourje is an existing run-time
verification library for communication protocols in Clojure, based on dynamic
MPST. The original version of Discourje can detect only safety violations. In
this paper, we present an extension of Discourje to detect also liveness
violations.
\end{abstract}


\section{Introduction}\label{sect:intr}

\subsubsection*{Background.}

With the advent of multicore processors, multithreaded program\-ming---a
notoriously error-prone en\-ter\-prise---has become increasingly important.

Because of this, mainstream languages have started to offer core support for
higher-level {\textit{communication} primitives} besides lower-level
{\textit{synchronisation} primitives} (e.g., Clojure, Go, Kotlin, Rust). The
idea has been to add \textit{message passing} as an abstraction on top of
\textit{shared memory}, for---supposedly---\textit{channels} are easier to use
than \textit{locks}. However, empirical research shows that, actually, ``message passing
does not necessarily make multithreaded programs less error-prone than shared
memory'' \cite{DBLP:conf/asplos/TuLSZ19}.
%
%
One of the core challenges is as follows: given a specification ${S}$ of the
\textit{communication protocols} that an implementation ${I}$ should fulfil, how to
prove that ${I}$ is \textit{safe} and \textit{live} relative to ${S}$? Safety
means that ``bad'' channel actions never happen: \underline{if} a channel action
happens in ${I}$, \underline{then} it is allowed to happen by ${S}$
(protocol compliance). Liveness means that ``good'' channel
actions eventually happen (communication deadlock freedom).

\subsubsection*{Multiparty session typing (MPST).}

MPST~\cite{DBLP:conf/popl/HondaYC08} is a formal method to automatically prove
safety and liveness of implementations relative to specifications. The idea is
to implement communication protocols as \textit{sessions} (of communicating
threads), specify them as \textit{behavioural types}
\cite{DBLP:journals/ftpl/AnconaBB0CDGGGH16,DBLP:journals/csur/HuttelLVCCDMPRT16}, and verify the former against the latter using behavioural type checking. Formally, the central theorem is that well-typed\-ness implies safety and liveness. Over the past fifteen years, much progress has been made, including the development of many tools to combine MPST with mainstream languages (e.g., F\#~\cite{DBLP:conf/cc/NeykovaHYA18}, F$^\star$~\cite{DBLP:journals/pacmpl/00020HNY20}, Go~\cite{DBLP:journals/pacmpl/CastroHJNY19}, Java~\cite{DBLP:conf/fase/HuY16,DBLP:conf/fase/HuY17}, OCaml~\cite{DBLP:conf/ecoop/ImaiNYY19}, Rust~\cite{DBLP:conf/coordination/LagaillardieNY20,DBLP:conf/ecoop/LagaillardieNY22}, Scala~\cite{DBLP:conf/ecoop/ScalasDHY17,DBLP:conf/ecoop/CledouEJP22,DBLP:conf/issta/FerreiraJ23,DBLP:conf/ecoop/BarwellHY023}, and TypeScript~\cite{DBLP:conf/cc/Miu0Y021}).

Behavioural type checking can be done \textit{statically} at compile-time or
\textit{dynamically} at run-time. The disadvantage of static MPST is, it is
conservative: statically checking \textit{each possible run} of a session is
often prohibitively complicated---if computable at all---so sessions are often
unnecessarily rejected. In contrast, the advantage of dynamic MPST is, it is
liberal: dynamically checking \textit{one actual run} of a session is much
simpler, so sessions are never unnecessarily rejected.


\subsubsection*{This work.}

\textit{Discourje} (pronounced ``discourse'')
\cite{DBLP:conf/tacas/HamersJ20,DBLP:conf/sigsoft/HorlingsJ21,DBLP:conf/isola/HamersJ20} is a library that adds dynamic MPST to \textit{Clojure}\footnote{A Lisp that runs on the JVM, with core support for channel-based message passing.}. It has a specification language to write behavioural types (em\-bedded as an internal DSL in Clojure) and a verification engine to dynamically type-check sessions against them. The \textbf{key design goals} have been to achieve high expressiveness (cf. static MPST) and to be particularly mindful of ergonomics (i.e., make Discourje's usage as frictionless as possible).

In a nutshell, at run-time, Discourje's dynamic type checker simulates
behavioural type $S$---as if it were a state machine---alongside session $I$.
Each time when a channel action is about to happen in $I$, the dynamic type
checker intervenes and first verifies if a corresponding transition can happen
in $S$. If so, both the channel action and the transition happen. If not, an
exception is thrown.

However, while safety violations are detected in this way (protocol
incompliance), liveness violations are not (communication deadlocks: threads
cyclically depend on each others' channel actions, and so, they collectively get
stuck). This is a serious limitation relative to static MPST. In this paper, we
present an extension of Discourje to detect also liveness violations. Achieving
this, \textit{without compromising the key design goals}, has been an elusive
problem that for years we did not know how to solve (e.g., we could not reuse
variants of existing techniques for static MPST at run-time, as this would
negatively affect expressiveness).

\autoref{sect:overv} of this paper demonstrates \underline{that} it can be done,
while \autoref{sect:details} outlines \underline{how}. The key idea is to use
``mock'' channels, which mimic ``real'' channels, to track ongoing
communications: before any channel action happens on a real channel, it is first
tried on a corresponding mock channel, allowing us to check if \textit{all}
threads would get stuck in a \textit{total} communication deadlock as a result.


\section{Demonstration}\label{sect:overv}

\begin{figure}[t]
	\textbf{Discourje:}
	\begin{itemize}
		\item \dcjm`(defthread $\mathit{id}$)`/\dcjm`(defsession $\mathit{id}$ [$\mathit{args}$] $\mathit{body}$)` specifies a thread name/protocol.
		
		\item \dcjm`(-->>/--> $t$ $p$ $q$)` specifies an asynchronous\slash
		synchronous communication of a value of data type $t$ through a buffered\slash
		unbuffered channel from $p$ to $q$.
		
		\item \dcjm`(alt ...)` and \dcjm`(cat/par ...)` specify choice and
		sequencing/interleaving.
		
		\item Names of threads and protocols are prefixed by an otherwise meaningless
		colon.
	\end{itemize}
	\textbf{Clojure:}
	\begin{itemize}
		\item \cljm`(thread $\mathit{body}$)`, \cljm`(chan)`, and %
		\cljm`(chan $\mathit{size}$)` implement the creation of new thread, a new
		unbuffered channel, and a new buffered channel.
		
		\item \cljm`(>!! $\mathit{ch}$ $\mathit{expr}$)` implements the send of the
		value of $\mathit{expr}$ through $\mathit{ch}$.
		
		\item \cljm`(<!! $\mathit{ch}$)` implements the receive of a value through
		$\mathit{ch}$.
		
		\item %
		\cljm`(alts!! [$\mathit{act}_1$ ... $\mathit{act}_n$])` implements a
		selection of one of the channel actions, depending on their dis/enabledness
		(cf. \texttt{select} of POSIX sockets and Go channels). If $\mathit{act}_i$ is
		a send, it is a pair \cljm`[$ch$ $v$]`; if it is a receive, it is just $ch$.
		The function returns a pair \cljm`[$v$ $ch$]` where $v$ is the value
		sent\slash received, and $ch$ is the channel.
	\end{itemize}
	
	\caption{Discourje and Clojure in a nutshell}
	\label{fig:dcj+clj}
\end{figure}

We demonstrate the extension to detect liveness violations with two examples.
For reference, \autoref{fig:dcj+clj} summarises the main elements of Discourje
and Clojure.

\begin{figure}[t]
	\begin{minipage}{.5\linewidth-2mm}
		\begin{lstlisting}[language=dcj]
(defthread :buyer1)
(defthread :buyer2)
(defthread :seller)

(defsession :two-buyer []
  (cat
   (-->> String :buyer1 :seller)
   (par
    (cat
     (-->> Double :seller :buyer1)
     (-->> Double :buyer1 :buyer2))
    (-->> Double :seller :buyer2))
   (-->> Boolean :buyer2 :seller)))
		\end{lstlisting}
		\subcaption{Specification in Discourje}
		\label{fig:twobuyer:dcj}
	\end{minipage}%
	\hfill%
	\begin{minipage}{.5\linewidth-2mm}
		\begin{minipage}{.5\linewidth-2mm}
			\begin{lstlisting}[language=clj]
(def c1 (chan 1))
(def c2 (chan 1))
(def c3 (chan 1))
(def c4 (chan 1))
(def c5 (chan 1))
(def c6 (chan 1))

(thread ;; Buyer1
  (>!! c1 "book")
  (let
    [x (<!! c5)
     y (/ x 2)]
    (>!! c2 y)))
			\end{lstlisting}
		\end{minipage}%
		\hfill%
		\begin{minipage}{.5\linewidth-2mm}
			\begin{lstlisting}[language=clj, firstnumber=last]
(thread ;; Buyer2
  (let 
    [x (<!! c6)
     y (<!! c2)
     z (= x y)]
    (>!! c4 z)))

(thread ;; Seller
  (<!! c1)
  (>!! c5 20.00)
  (>!! c6 20.00)
  (println
    (<!! c4)))
			\end{lstlisting}
		\end{minipage}
		\subcaption{Implementation in Clojure}
		\label{fig:twobuyer:clj}
	\end{minipage}
	
\medbreak

	To dynamically type-check the session, the following code creates a
	\textit{monitor} for the session, and \textit{links} it to each channel along
	with the intended sender and receiver:

\medbreak

	\begin{lstlisting}[language=clj, numbers=none]
(def m (monitor :two-buyer :n 3))
(link c1 :buyer1 :buyer2 m) (link c2 :buyer1 :seller m)
(link c3 :buyer2 :buyer1 m) (link c4 :buyer2 :seller m)
(link c5 :seller :buyer1 m) (link c6 :seller :buyer2 m)
	\end{lstlisting}
	
	\caption{Two-Buyer (\autoref{exmp:twobuyer})}
	\label{fig:twobuyer}
\end{figure}

\begin{example}\label{exmp:twobuyer}
	The \textit{Two-Buyer} protocol consists of \textit{Buyer1}, \textit{Buyer2},
	and \textit{Seller} \cite{DBLP:conf/popl/HondaYC08}: ``Buyer1 and Buyer2 wish
	to buy an expensive book from Seller by combining their money. Buyer1 sends the
	title of the book to Seller, Seller sends to both Buyer1 and Buyer2 its quote,
	Buyer1 tells Buyer2 how much she can pay, and Buyer2 either accepts the quote
	or rejects the quote by notifying Seller.''
	
	\autoref{fig:twobuyer} shows a behavioural type and a session. It is safe and
	live. In contrast, if we had accidentally written \cljm`(<!! c3)` on line 11
	(i.e., Buyer1 tries to receive from Buyer2 instead of Seller), then it
	deadlocks. The original Discourje does not detect this liveness violation, but
	with the extension, an exception is thrown. \qed
\end{example}

\begin{figure}[t]
	\begin{minipage}{.4\linewidth-2mm}
		\begin{lstlisting}[language=dcj]
(defthread :c) (defthread :s1)
(defthread :b) (defthread :s2)

(defsession :load-balancer []
  (cat 
   (-->> Long :c :b)
   (alt
    (cat
     (-->> Long :b  :s1)
     (-->  Long :s1 :c))
    (cat
     (-->> Long :b  :s2)
     (-->  Long :s2 :c)))))
		\end{lstlisting}
		\subcaption{Specification in Discourje}
		\label{fig:lb:dcj}
	\end{minipage}%
	\hfill%
	\begin{minipage}{.6\linewidth-2mm}
		\begin{minipage}{.5\linewidth-2mm}
			\begin{lstlisting}[language=clj]
(def c1 (chan))
(def c2 (chan))
(def c3 (chan))
(def c4 (chan 512))
(def c5 (chan 1024))

;; Load Balancer
(thread
  (let [x (<!! c1)]
    (alts!!
      [[c4 x]
       [c5 x]])))
$ $
			\end{lstlisting}
		\end{minipage}%
		\hfill%
		\begin{minipage}{.5\linewidth-2mm}
			\begin{lstlisting}[language=clj, firstnumber=last]
(thread ;; Client
  (>!! c1 5)
  (alts!! [c2 c3])))

(thread ;; Server1
  (let [x (<!! c2)
        y (inc x)]
    (>!! c2 y))))

(thread ;; Server2
  (let [x (<!! c3)
        y (inc x)]
    (>!! c3 y))))
			\end{lstlisting}
		\end{minipage}
		\subcaption{Implementation in Clojure}
		\label{fig:lb:clj}
	\end{minipage}
	
\medbreak
	
	To dynamically type-check the session:
	
\medbreak
	
	\begin{lstlisting}[language=clj, numbers=none]
(def m (monitor :load-balancer :n 4)) (link c4 :b :s1 m) (link c2 :s1 :c m)
                    (link c1 :c :b m) (link c5 :b :s2 m) (link c3 :s2 :c m)
	\end{lstlisting}
	
	\caption{Load Balancing (\autoref{exmp:lb}}
	\label{fig:lb}
\end{figure}

\begin{example}\label{exmp:lb}
	The \textit{Load Balancing} protocol consists of \textit{Client},
	\textit{Server1}, \textit{Server2}, and \textit{LoadBalancer}. First, a request
	is communicated synchronously from Client to LoadBalancer, and asynchronously
	from LoadBalancer to  Server1 or Server2. Next, the response is communicated
	synchronously from that server to Client.
	
	\autoref{fig:lb} shows a behavioural type and a session. It is safe but not
	live. There are two deadlocks. The first one occurs because Server1 and Server2
	try to receive from \clj`c2` and \clj`c3` on lines 19 and 23; this should be
	\clj`c4` and \clj`c5`. The second deadlock occurs because one of the servers
	will never receive a value and, as a result, block the entire program from
	terminating. The original Discourje does not detect these liveness violations,
	but with the extension, exceptions are thrown. \qed
\end{example}


\section{Technical Details}
\label{sect:details}

\subsubsection*{Requirements.}

In this section, we outline how the extension to detect liveness violations
works, focussing on the core deadlock detection algorithm. We begin by stating
the rather complicated requirements for this algorithm, as entailed by
Discourje's key design goals regarding expressiveness and ergonomics
(\autoref{sect:intr}):

\begin{itemize}
	\item \textbf{Expressiveness:} The algorithm must be applicable to any
	combination of buffered and unbuffered channels, and to all functions \clj`>!!`
	(send), \clj`<!!` (receive), and \clj`alts!!` (select). Thus, the programmer
	can continue to freely mix synchronous and asynchronous sends/receives, possibly
	selected dynamically.
	
	\item \textbf{Ergonomics:} The algorithm must call only into the public API of
	Clojure's standard libraries, without modifying the internals, and without
	relying on JVM interoperability. Thus, the programmer can write portable code
	that runs on different versions of Clojure and on different architectures.
\end{itemize}

\noindent The combination of these requirements has made the design of the
algorithm elusive. For instance, the expressiveness requirement means that we
cannot simply reuse existing distributed algorithms for deadlock detection
(e.g.,
\cite{DBLP:journals/dpd/SrinivasanR11,DBLP:conf/sc/HilbrichSNPBM13,DBLP:journals/dc/BrachaT87,DBLP:journals/vldb/KrivokapicKG99}), as they typically do not support mixing of synchrony and asynchrony. The ergonomics requirement means that we cannot instrument Clojure's internal code to manage threads, nor can we use Java's thread monitoring facilities.

\subsubsection*{Terminology.}

A \textit{channel action} is either a \textit{send} of $v$ through $ch$,
represented as \cljm`[$ch$ $v$]`, or a \textit{receive} through channel $ch$,
represented as just $ch$ (cf. \clj`alts!!` in \autoref{fig:dcj+clj}). A channel
action is \textit{pending} if it has been initiated but not yet completed. A pending
channel action is either \textit{enabled} or \textit{disabled}, depending on
$ch$:
\begin{itemize}
	\item when $ch$ is a buffered channel, a pending send is enabled iff $ch$ is
	non-full, while a pending receive is enabled iff $ch$ is non-empty;
	
	\item when $ch$ is an unbuffered channel, a pending send is enabled iff a
	corresponding receive is pending, and vice versa.
\end{itemize}
When a thread initiates channel actions, but they are disabled, it is
\textit{suspended}. When a disabled channel action becomes enabled, the suspended
thread is \textit{resumed}. A \textit{communication deadlock} is a situation where
each thread is suspended.

\subsubsection*{Setting the stage.}

Normally, channel actions are initiated via functions \clj`>!!`, \clj`<!!`, and
\clj`alts!!`. When these functions are called using the extension, the dynamic
type checker intervenes and first calls %
\cljm`(detect-deadlocks [$\mathit{act}_1$ ... $\mathit{act}_n$])` to initiate
corresponding ``mock'' channel actions on ``mock'' channels. Each mock channel
mimics a ``real'' channel and is used only by the dynamic type checker.

The mock channels have the same un/buffered properties and contents as the real
channels, except that values are replaced with tokens. So, if
\cljm`detect-deadlocks` detects a deadlock on the mock channels, then a deadlock
will occur on the real channels, too. (Mock channels are also essential to
detect safety violations.)

To initiate the mock channel actions, a separate function in the public API of
Clojure's standard libraries is used: %
\cljm`(do-alts $\mathit{f}$ $\mathit{acts}$ $\mathit{config}$)`. It resembles
\clj`alts!!`, except that it never suspends the calling thread. Instead, a call
of \clj`do-alts` immediately returns and, asynchronously, initiates the channel
actions in $\mathit{acts}$ and calls $\mathit{f}$ when one is completed. In this
way, initiation of mock channel actions can be decoupled from suspension of
threads (demonstrated below).

\subsubsection*{Algorithm.}

Let \clj`n` be the number of threads. The idea to detect deadlocks is to
identify the situation when \clj`n-1` threads are already suspended, while the
\textbf{last thread} is \textbf{about to be suspended}. In that situation,
instead of suspending the last thread, an exception is thrown to flag the
liveness violation. In code:

\medbreak\begin{lstlisting}[language=clj]
(defn detect-deadlocks [mock-acts] ;; $\color{gray}\mathit{act}_1$ ... $\color{gray}\mathit{act}_n$
  (let [ret (about-to-be-suspended? mock-acts)]
    (if (true? ret)
      (let [ret (last-thread? mock-acts)]
        (if (true? ret) (throw (ex-info "deadlock!" {})) ret)) ret)))
\end{lstlisting}\medbreak

Function \clj`about-to-be-suspended?` checks if any of the \clj`mock-acts` is
enabled. If so, it immediately initiates and completes it, and returns the
result (of the form \cljm`[$v$ $ch$]`). If not, the function returns \clj`true` to
indicate that the current thread would indeed be suspended if \clj`mock-acts`
were to be initiated. In code:

\medbreak\begin{lstlisting}[language=clj, firstnumber=6]
(defn about-to-be-suspended? [mock-acts]
  (let [ret @(do-alts (fn [_] nil) mock-acts {:default nil})]
    (if (not= ret [nil :default]) ret true)))
\end{lstlisting}\medbreak

\noindent On line 7, optional parameter \clj`{:default nil}` configures
\clj`alts!!` such that it immediately returns \clj`[nil :default]` when all
\clj`mock-acts` are disabled.

Function \clj`last-thread?` increments the number of suspended threads and
checks if the number is less than \clj`n`. If so, it initiates \clj`mock-acts`,
and actually suspends the current thread. If not, the function returns
\clj`true` to indicate that the current thread is indeed the last one, so a
deadlock is detected. In code:
\medbreak\begin{lstlisting}[language=clj, firstnumber=9]
(def i (atom 0)) ;; number of suspended threads (private to the algorithm)

(defn last-thread? [mock-acts]
  (if (< (swap! i inc) n)  ;; increment `i` (`swap!` returns the new value)
    (let [p (promise)]     ;; create promise to store result of `mock-acts`
      (do-alts (fn [x] (deliver p x)) mock-acts {}) ;; initiate `mock-acts`,
                           ;; and store result `x` of one of them in `p`
                           ;; upon completion, all asynchronously
      (let [ret (deref p)] ;; suspend thread (`deref` blocks until `deliver`)
        (swap! i dec)      ;; decrement `i`
        ret)
    true))
\end{lstlisting}\medbreak
The code shown so far explains the general idea behind the algorithm. However,
the details are more involved: our presentation does not yet account for data
races, several of which are possible. For instance, suppose that there are two
threads (Alice and Bob), that they initiate corresponding channel actions (no
deadlock), and that calls of \clj`detect-deadlocks` are scheduled as follows:
\begin{quote}
	\textbf{(1)} Alice executes \clj`about-to-be-suspended?`. It returns \clj`true`.
	\textbf{(2)} Bob executes \clj`about-to-be-suspended?`. It, again, returns
	\clj`true`, as Alice has not yet executed \clj`last-thread?`. \textbf{(3)} Bob
	executes \clj`last-thread?`. It increments \clj`n` to \clj`1` and suspends Bob.
	\textbf{(4)} Alice executes \clj`last-thread?`. It increments \clj`n` to
	\clj`2`, detects that Alice is last, and immediately returns \clj`nil`.
\end{quote}

\noindent At this point, mistakenly, an exception is thrown. There are more
subtle data races, too. The core issue is that \clj`about-to-be-suspended?` and
\clj`last-thread?` should be run \textit{atomically} to avoid problematic
schedules (e.g., the one above). Details appear in
\ifdefined\tr{\autoref{sect:algo}}\else{the technical report \cite[Sect.
A]{tr}}\fi. The actual source code was validated using both unit tests and
whole-program tests.


\section{Conclusion}

Closest to the work in this paper is existing work on dynamic MPST
\cite{DBLP:journals/fac/NeykovaBY17,DBLP:journals/tcs/BocchiCDHY17,DBLP:conf/rv/HeuvelPD23,DBLP:conf/cc/NeykovaY17,DBLP:conf/cc/NeykovaHYA18} and alternate forms of dynamic behavioural typing \cite{DBLP:conf/ecoop/BurloFS21,DBLP:journals/pacmpl/MelgrattiP17,DBLP:journals/jlap/GommerstadtJP22,DBLP:journals/scp/BurloFSTT22}. However, none of these tools can check for liveness at run-time. Also closely related is existing work on dynamic deadlock detection in distributed systems (e.g., \cite{DBLP:journals/dpd/SrinivasanR11,DBLP:conf/sc/HilbrichSNPBM13,DBLP:journals/dc/BrachaT87,DBLP:journals/vldb/KrivokapicKG99}). However, as stated in \autoref{sect:details}, these algorithms do not fit our requirements. Finally, we are aware of only two other works that use formal techniques to reason about Clojure programs: the formalisation of an optional type system for Clojure  \cite{DBLP:conf/esop/Bonnaire-Sergeant16}, and a translation from Clojure to Boogie \cite{DBLP:conf/fmco/BarnettCDJL05,DBLP:conf/tools/PinzaruR19}.

As next steps, we aim to extend Discourje with support for Clojure's built-in
publish--subscribe mechanism, as well as its channel composition facilities. We
also aim to quantitatively measure, compare, and optimise the performance
overhead of our deadlock detection algorithm.

\clearpage
\begin{credits}
\subsubsection*{\discintname}
The author has no competing interests to declare that are relevant to the
content of this article.
\end{credits}

\bibliographystyle{splncs04}
\bibliography{fm2024}

\ifdefined\tr

\clearpage
\appendix

\section{Details of the Algorithm: Synchronisation}
\label{sect:algo}

\begin{figure}[p]
	\begin{lstlisting}[language=clj]
;; numbers of threads
(def n ...)      ;; total
(def i (atom 0)) ;; suspended

;; channel-based lock
(def semaphore (chan 1))
(def permit "")
(>!! semaphore permit)

;; channel-based barriers
;; (def barriers (atom {})) $\label{shadowbegin1}$
;; (defn install-barrier [barriers mock-acts]
;;   (let [barrier (a/chan)
;;         chs (mapv #(if (vector? %) (first %) %) mock-acts)]
;;     (swap! barriers conj [barrier chs])
;;     barrier))
;; (defn uninstall-barrier [barriers ch]
;;   (loop []
;;     (let [pred (fn [[_ chs]] (some #{ch} chs))
;;           old @barriers
;;           new (filterv (complement pred) old)]
;;       (if (compare-and-set! barriers old new)
;;         (first (first (filterv pred old)))
;;         (recur)))))$\label{shadowend1}$

(defn detect-deadlocks [mock-acts]
  (<!! semaphore)
  ;; --- BEGIN CRITICAL SECTION --- $\label{csbegin}$
  
  ;; `about-to-be-suspended?`
  (let [[v ch] @(do-alts (fn [_] nil) mock-acts {:default nil})] $\label{alts1}$
    (if (= ch :default)
    
      ;; `last-thread?`
      (if (< (swap! i inc) n)
        (let [;; barrier (install-barrier mock-acts) $\label{shadow}$
              p (promise)]
          (do-alts (partial deliver p) mock-acts {}) $\label{alts2}$
          
          ;; --- END CRITICAL SECTION --- $\label{csend1}$
          (>!! semaphore permit)
          
          (let [[v ch] (deref p)]
            (swap! i dec)
            ;; (a/>!! barrier "") $\label{shadowbegin2}$
             
            ;; Return
            [v ch]))
        (do
        
          ;; --- END CRITICAL SECTION --- $\label{csend2}$
          (>!! semaphore permit)
          
          ;; Throw
          (throw (ex-info "deadlock" {}))))
      (do
        ;; (if-let [barrier (uninstall-barrier barriers ch)] $\label{shadowbegin3}$
        ;;   (a/<!! barrier)) $\label{shadowend3}$
        
        ;; --- END CRITICAL SECTION --- $\label{csend3}$
        (>!! semaphore permit)
        
        ;; Return
        [v ch]))))
	\end{lstlisting}
	
	\caption{Full algorithm}
	\label{fig:algo}
\end{figure}

\autoref{fig:algo} shows our full deadlock detection algorithm,\footnote{The
actual source code is even more complex, as it also needs to account for closing
of channels and termination of threads; we skip those complications here.} the
core of which was presented in \autoref{sect:details}. We inlined the bodies of
functions \clj`about-to-be-suspended?` and \clj`last-thread?` to make the
presentation below simpler. For now, we ignore the commented lines
\ref{shadowbegin1}--\ref{shadowend1}, line \ref{shadow}, line
\ref{shadowbegin2}, and lines \ref{shadowbegin3}--\ref{shadowend3}; we will
explain their purpose later in this section.

Relative to \autoref{sect:details}, the main change is that \autoref{fig:algo}
also contains additional synchronisation to protect the \textit{critical
section} of the algorithm (i.e., execute the bodies of
\clj`about-to-be-suspended?` and \clj`last-thread?` atomically). The critical
section is entered at line \ref{csbegin}; it is exited at line \ref{csend1}, at
line \ref{csend2}, and at line \ref{csend3}. To achieve protection, first, on
lines 6--8, we define a helper channel, we define a special value, and we
initialise the helper channel to contain the special value. Next, we add
\clj`(<!! semaphore)` (acquire the permit) at the entry of the critical section,
while we add %
\clj`(>!! semaphore permit)` (release the permit) at each exit. As a result,
only one thread can run the critical section at a time.

Protection of the critical section already fixes most data races, but there is
one final issue remaining. To exemplify this, suppose that there are two
threads, Alice (the sender) and Bob (the receiver), that they initiate two pairs
of corresponding synchronous mock channel actions on an unbuffered mock channel
(no deadlock), and that calls of \clj`detect-deadlocks` are scheduled as follows:

\begin{itemize}
	\item First, Alice enters the critical section to initiate her first send. As
	Bob has not yet initiated his corresponding receive, Alice's call of
	\clj`do-alts` on line \ref{alts1} immediately returns \clj`[nil :default]`.
	
	\item Next, Alice increments \clj`n` to \clj`1`, detects that she is not last,
	and actually initiates her first send by calling \clj`do-alts` on line
	\ref{alts2}.
	
	\item Next, Alice exits the critical section on line \ref{csend1}.
	
	\item Next, Alice calls \clj`deref` and is suspended until her first send is
	completed.

\medbreak

	\item Next, Bob enters the critical section to initiate his first receive. As
	Alice has already initiated her corresponding send, Bob's call of \clj`do-alts`
	on line \ref{alts1} immediately initiates and completes the receive. As the
	mock channel is unbuffered, at the same time, Alice's send is completed as
	well.

	\item Next, Bob exits the critical section on line \ref{csend3}.
	
\medbreak
	
	\item Next, as Alice's mock channel action was completed, Alice resumes.
	However, \textit{she does not get to decrement \clj`i` yet}. Instead, Bob is
	scheduled to go in between. From this point onwards, failure is impending:
	neither Alice nor Bob is suspended, but \clj`i` equals \clj`1`.
	
\medbreak
	
	\item Next, Bob re-enters the critical section to initiate his second send. As
	Alice has not yet initiated her corresponding receive, Bob's call of
	\clj`do-alts` on line \ref{alts1} immediately returns \clj`[nil :default]`.
	
	\item Next, Bob increments \clj`n` to \clj`2`, detects that he is last, and
	throws an exception.
\end{itemize}

\noindent However, the throw of an exception is a mistake: there is no deadlock,
as Alice is already resumed. In general, the root of the problem is that when
one thread $A$ completes a mock channel action, as a result, a mock channel
action of another thread $B$ may be completed as well, causing $B$ to resume;
however, the value of \clj`i` does not reflect the resumption yet, so if $A$
races to initiate another mock channel action, it might observe the obsolete
value of \clj`i`.

To address this issue, additional synchronisation between $A$ and $B$ is needed
to ensure that $A$ exits the critical section only after $B$ has decremented
\clj`i`. To achieve this, we uncomment lines
\ref{shadowbegin1}--\ref{shadowend1}, line \ref{shadow}, line
\ref{shadowbegin2}, and lines \ref{shadowbegin3}--\ref{shadowend3}:

\begin{itemize}
	\item On lines \ref{shadowbegin1}--\ref{shadowend1}, we define an initially
	empty map of unbuffered \textit{barrier channels} and two functions to
	de/populate it. Function \clj`install-barrier` associates a new barrier channel
	with a vector of the channels that occur in \clj`mock-acts`, while function
	\clj`uninstall-barrier` disassociates an existing barrier channel that has
	\clj`ch` in its image.
	
	\item On line \ref{shadow}, right before the \clj`mock-acts` are initiated, the
	map is populated.
	
	\item On line \ref{shadowbegin2}, right after one of the \clj`mock-acts` is
	completed, the associated barrier channel is used for synchronisation (i.e.,
	send an arbitrary value).
	
	\item On lines \ref{shadowbegin3}--\ref{shadowend3}, right after one of the
	\clj`mock-acts` is completed, the associated barrier channel (if any) is used
	for synchronisation (i.e., receive an arbitrary value), and the map is
	depopulated.
\end{itemize}

\noindent Thus, right before a thread $A$ exits the critical section on line
\ref{csend3}, after completing a mock channel action on \clj`ch`, it checks if a
\clj`barrier` is associated with \clj`ch`. If so, then it must be the case that:
\textbf{(1)} a thread $B$ has previously initiated a mock channel action on
\clj`ch`, but at that time, it was disabled, so $B$ associated \clj`barrier`
with \clj`ch`; \textbf{(2)} $B$ was suspended; \textbf{(3)} because $A$ has now
completed a corresponding mock channel action on \clj`ch`, $B$'s mock channel
action on \clj`ch` is now enabled; \textbf{(3)} $B$ is resumed. Thus, $A$ needs
to wait until $B$ has decremented \clj`i` by using the \clj`barrier`.
Reciprocally, $B$ also uses the \clj`barrier` after decrementing \clj`i`. As a
result, when $A$ exits the critical section after synchronising with $B$ through
\clj`barrier`, \clj`i` has the right value.

We note that we assume that $A$ and $B$ perform corresponding mock channel
actions on \clj`ch`. Satisfaction of this assumption can also be checked at
run-time: we just need to make sure that if $A$ and $B$ are different threads,
then they never both send to, or both receive from, the same channel. This is
straightforward to check using Discourje (because we already know the intended
sender and the intended receiver of each channel, as demonstrated at the bottom
of \autoref{fig:twobuyer}), and it should always be the case.

We are now in a position to prove the main correctness results. We begin by
stating some basic facts about the algorithm.

\begin{proposition}
	At most one thread is in the critical section at a time.
\end{proposition}

\begin{proposition}
	A thread can become about to be suspended only in the critical section.
\end{proposition}

\begin{proposition}
	A thread can be resumed only when another thread is in the critical section and
	completes a mock channel action by calling \clj`do-alts` on line \ref{alts1}.
\end{proposition}

We now formulate the crucial \textit{critical section invariant}.

\begin{lemma}[critical section invariant]\label{lem:1}
	If the number of threads that are, or are about to be, suspended is exactly
	\clj`i` before executing the critical section, then this property also holds
	after, unless an exception is thrown.
\end{lemma}

\begin{proof}
	Suppose that thread $A$ enters the critical section. There are two cases.
	
	If none of the \clj`mock-acts` are enabled, then $A$ will increment \clj`i`,
	initiate the \clj`mock-acts`, and exit the critical section; at this point, $A$
	is about to be suspended (when it subsequently calls \clj`deref`). By the
	previous propositions, no other thread can become about to be suspended, or be
	resumed, while $A$ is in the critical section (but threads that are already
	about to be suspended can be suspended), so when $A$ exits it, the incremented
	\clj`i` properly reflects that $A$ is about to be suspended, in addition to the
	threads that already were so upon entry.
	
	If at least one of the \clj`mock-acts` is enabled, $A$ will immediately
	initiate and complete it. If another thread $B$ can be resumed because of this,
	then $A$ will synchronise with $B$ after $B$ has decremented \clj`i` before
	exiting the critical section. Along the same lines as in the previous case,
	\clj`i` properly reflects either that no additional threads are (about to be)
	suspended\slash resumed (if there is no $B$), or that $B$ has been resumed,
	when $A$ exits the critical section. \qed
\end{proof}

Trivially, the first time the critical section is entered, the number of
suspended threads is exactly \clj`i` (namely \clj`0`). Thus, inductively, the
critical section invariant implies that always when the critical section is
entered, the number of (about to be) suspended threads is exactly \clj`i`. We
now formulate the main theorem.

\begin{theorem}\label{thm:1}
	An exception is thrown if, and only if, there is a deadlock.
\end{theorem}

\begin{proof}
	\strut
	\begin{itemize}
		\item[($\Leftarrow$)] Suppose that there is a deadlock. Then, all threads are
		suspended. Then, due to the critical section invariant, \clj`i` must be
		\clj`n`. Then, there must have been a thread that incremented \clj`i` to
		\clj`n`. Then, that thread threw an exception.
	
		\item[($\Rightarrow$)] Suppose that an exception is thrown by a thread.
		Then, \clj`i` was incremented to \clj`n` by that thread. Then, \clj`do-alts`
		on line \ref{alts1} returned \clj`[nil :default]`. Then, none of the
		\clj`mock-acts` were enabled, and as a result, the thread would have been
		suspended. Furthermore, because \clj`i` was incremented to \clj`n` during the
		critical section, \clj`i` must have been \clj`n-1` before the critical
		section. Then, by the critical section invariant, the number of threads that
		were, or were about to be, suspended upon entry of the critical section was
		\clj`n-1`. As the current thread would have become about to be suspended as
		well, \clj`n` threads would have been suspended, or about to become so, upon
		exit of the critical section. Thus, there is a deadlock. \qed
	\end{itemize}
\end{proof}

\fi

\end{document}